\documentclass[12pt, a4paper,reqno]{amsart}

\usepackage{geometry}
\usepackage{pgf,tikz}
\usepackage{mathrsfs}
\usetikzlibrary{arrows}
    
\usepackage{color} \definecolor{bleu_sombre}{rgb}{0,0,0.6}  
\definecolor{rouge_sombre}{rgb}{0.8,0,0}
\definecolor{vert_sombre}{rgb}{0,0.6,0}
\usepackage[plainpages=false,colorlinks,linkcolor=bleu_sombre,citecolor=rouge_sombre,urlcolor=vert_sombre,breaklinks]{hyperref}

\usepackage[english]{babel}

\usepackage{amsmath,amssymb,amsthm,graphicx,amsfonts,url,color,enumerate,dsfont,stmaryrd,mathabx, csquotes}

\theoremstyle{plain}
\newtheorem{theorem}{{Theorem}}[section] 
\newtheorem*{theorem*}{{Theorem}}
\newtheorem{proposition}[theorem]{Proposition}
\newtheorem*{proposition*}{Proposition}

\newtheorem*{corollary*}{Corollary}
\newtheorem{lemma}[theorem]{Lemma}
\newtheorem*{lemma*}{Lemma}

\theoremstyle{definition}

\newtheorem*{definition*}{Definition}

\theoremstyle{remark}
\newtheorem{remark}[theorem]{Remark}

\makeatletter

\@addtoreset{equation}{section}  
\makeatother

\newcommand {\limt}[2]{\xrightarrow[#1 \to #2]{}}

\newcommand{\abs}[1]{\left\vert #1\right\vert}      
\newcommand{\nr}[1]{\left\Vert #1\right\Vert}         
\newcommand{\innp}[2]{\left< #1 , #2 \right>}          
\newcommand{\pppg}[1] {\left< #1 \right>}
\newcommand{\set}[1]{\left\{ #1 \right\}}		

\renewcommand{\leq}{\leqslant}	\renewcommand{\geq}{\geqslant}

\newcommand{\inv}{^{-1}}

\newcommand{\st}{\,:\,}
         
\renewcommand{\Re}{\mathsf{Re}}        
  
\newcommand{\trsp}{^{\intercal}}

\newcommand{\Dom}{\mathsf{Dom}}
\newcommand{\Sp}{\sigma}

\newcommand{\R}{\mathbb{R}}		\newcommand{\C}{\mathbb{C}}
\newcommand{\N}{\mathbb{N}}

\renewcommand{\a}{\alpha}\newcommand{\g}{\gamma}\newcommand{\e}{\varepsilon} \renewcommand{\th}{\theta}\newcommand{\Th}{\Theta}\renewcommand{\k}{\kappa}\renewcommand{\l}{\lambda}\newcommand{\x}{\xi}\newcommand{\s}{\sigma}\newcommand{\vf}{\phi}\newcommand{\p}{\psi}\renewcommand{\o}{\omega}\renewcommand{\O}{\Omega}

\newcommand{\Sc}{{\mathcal S}}

\usepackage{mathrsfs}

\begin{document}

\renewcommand{\labelitemi}{---}

\newcommand{\ess}{{\mathsf{ess}}}

\newcommand{\LL}{\mathscr L}
\newcommand{\hLL}{\mathscr L_\x}
\newcommand{\LLh}{\mathscr L_h}
\newcommand{\hLLh}{\hat {\mathscr L}_h}

\newcommand{\etah}{\eta_h}

\title[Absence of embedded eigenvalues]{Absence of embedded eigenvalues for translationally invariant magnetic Laplacians}

\author{N. Raymond}
\address[N. Raymond]{Universit\'e d'Angers, CNRS, LAREMA - UMR 6093, 49045 Angers Cedex 01, France}
\email{nicolas.raymond@univ-angers.fr}

\author{J. Royer}
\address[J. Royer]{Institut de math\'ematiques de Toulouse - UMR 5219, Universit\'e de Toulouse, CNRS, 31062 Toulouse cedex 9, France}
\email{julien.royer@math.univ-toulouse.fr}

\begin{abstract}
 Translationnally invariant bidimensional magnetic Laplacians are considered. Using an improved version of the harmonic approximation, we establish the absence of point spectrum under various assumptions on the behavior of the magnetic field.
\end{abstract}

\maketitle

\section{Context and results}

\subsection{Translationally invariant magnetic Laplacians}
This paper is devoted to the description of the point spectrum of translationally invariant magnetic Laplacians in two dimensions. Here the magnetic field $B$ is assumed to be a smooth enough function that only depends on its first variable. More precisely, we assume that
\[\forall (x,y)\in \R^2\,,\quad B(x,y)=b(x),\]
where $b\in\mathscr{C}^1(\R,\R)$. 
Associated with $B$, we may consider a vector potential $\mathbf{A}=(A_{1}, A_{2})$ where
\begin{equation} \label{def-a}
A_{1}(x,y)=0\,,\qquad A_{2}(x,y)=a(x):= a_0 + \int_{0}^x b(u)\mathrm{d} u\,,
\end{equation}
for some arbitrary $a_0$. When the limits exist in $\R \cup \{\pm \infty\}$, we set 
\begin{equation} \label{def-phi}
\phi_\pm = \lim_{x \to \pm \infty} a(x).
\end{equation}

The magnetic Laplacian under consideration in this paper is the self-adjoint differential operator
\begin{equation} \label{def-LL}
\LL =(-i\nabla-\mathbf{A})^2=D_x^2 + \big( D_y - a(x) \big)^2\,,\quad D=-i\partial\,,
\end{equation}
equipped with the domain
\[\Dom(\LL) = \set{ u \in H^1_{\mathbf{A}}(\R^2) \st (-i\nabla-\mathbf{A})^2u\in L^2(\R^2)}\,,\]
where 
\[H^1_{\mathbf{A}}(\R^2)=\{u\in L^2(\R^2) \st (-i\nabla-\mathbf{A})u\in L^2(\R^2)\}\,.\]

\subsection{Context and motivation}
Due to the translation invariance, it is easy to see that the spectrum of $\mathscr{L}$ is essential:
\[\s(\mathscr{L})=\s_\ess(\mathscr{L})\,.\]

 The main question addressed in this paper is to find conditions under which $\mathscr{L}$ has no eigenvalue. Thus, we would like to exclude the existence of $(\lambda,\psi)\in[0,+\infty)\times\Dom(\mathscr{L})$ such that $\psi \neq 0$ and $\mathscr{L}\psi=\lambda\psi$. In order to understand how subtle this question can be, let us remark the following:
 \begin{itemize}
 	\item When $b$ is constant and non-zero, it is well-known that the spectrum is made of infinitely degenerate eigenvalues, the Landau  levels:
 	\[\s(\mathscr{L})=\{(2n-1)|b|\,,n\geq 1\}\,.\]
 	\item When $\phi_+$ or $\phi_-$ is finite, one will see in our proofs that
 	\[\s(\mathscr{L})=[0,+\infty)\,.\]
 	\end{itemize}
Thus, as noticed in the seminal paper \cite{I85}, even the nature of the essential spectrum itself strongly depends on the variations of $b$. 

In this paper, we focus our investigation on proving the \emph{absence of point spectrum}, even if, in some particular situations, our proof might also imply the absolute continuity of the spectrum. In particular, in Theorem \ref{th-semi-confined}, one will see that, if $b(x)$ behaves like $x^\alpha$ (with $\alpha\neq 0$ and $\alpha>-1$) at infinity, the Landau levels structure is lost as well as the existence of eigenvalues. Theorem \ref{th-a-bounded} is of asympotic nature: when $b\in L^1(\R,\R_+)$ and when the magnetic field is large, we show that the only possible eigenvalues are essentially of the order of the flux squared. 

Our main results deal with cases when $a$ is semi-bounded, semi-unbounded, and when $a$ is bounded. They partially extend  the results in \cite{I85} (where the assumptions imply $\displaystyle{\lim_{x\to\pm\infty} a(x)=\pm\infty}$) by considering non-necessarily bounded magnetic fields. 

More generally, this paper can be considered as an exploration of the conjecture stated in \cite[Theorem 6.6 \& Remark 1]{cycon}. Let us recall a theorem whose proof may be deduced from the investigation in \cite{I85} (and also \cite[Theorem 6.6]{cycon} where the magnetic field is allowed to vanish).
\begin{theorem}[Ywatsuka '85] \label{th-conf-semibounded}
	Assume
	\begin{enumerate}[\rm (i)]
		\item\label{eq.thm1i} either that (see \eqref{def-phi})
		\[
		\phi_- = \phi_+ = -\infty \quad \text{or} \quad  \phi_- = \phi_+ = +\infty\, ,
		\]
		\item\label{eq.thm1ii} or that $\displaystyle{\lim_{x\to\pm\infty} b(x)=b_{\pm}}$ with $b_{\pm}\in\R\setminus\{0\}$ distinct.
	\end{enumerate}
	Then $\LL$ has absolutely continuous spectrum. In particular, $\mathscr{L}$ has no eigenvalue.
\end{theorem}
\subsection{Some relations with the literature}
In \cite{I85}, the author is mainly concerned by proving the absolute continuity of the spectrum. Note that this issue is closely connected to the existence of edge currents (quantified by Mourre estimates), as explained for instance in \cite{HS15}, where positive magnetic fields are considered. The reader might also want to consider
\begin{itemize}
\item the physical considerations in \cite{PeetRej}, 
\item the paper \cite{HPPR16} considering the dispersion curves associated with non-smooth magnetic fields,
\item the contribution \cite{Tusek16} generalizing Iwatsuka's result by adding a translationnaly invariant electric potential, 
\item the paper \cite{Y08} devoted to dimension three and fields having cylindrical and longitudinal symmetries,
\item or \cite{MP18} where various estimates of the band functions are established for increasing, positive, and bounded magnetic fields, and applied to the estimate of quantum currents.
\end{itemize}

\subsection{Main results}
Let us now state our main theorems. In the first result we generalize Theorem \ref{th-conf-semibounded}.\eqref{eq.thm1ii} by considering situations where $\vf_+ = +\infty$ and $\vf_- \in \R \cup \{-\infty\}$.

\begin{theorem} \label{th-semi-confined}
 Assume that $\vf_-$ exists as an element of $\R\cup\{\pm\infty\}$ and that for some $\alpha \in (-1,0) \cup (0,+\infty)$ and $c_{1}, C>0$ we have 
\begin{equation} \label{hyp-a-asymp}
b(x)\underset{x\to+\infty}{\sim} c_{1} x^{\alpha} \quad \text{and} \quad \abs{b'(x)} \leq C \pppg x^{\alpha-1}.
\end{equation}
Then $\LL$ has no eigenvalue. 
\end{theorem}
\begin{remark}
\begin{itemize}
\item By the symmetry $x\mapsto -x$, we can easily adapt this theorem to consider behaviors in $-\infty$. We have a similar result if $-b$ satisfies \eqref{hyp-a-asymp}.
\item Theorem \ref{th-semi-confined} can be applied, for instance, to $b_\pm(x)=\langle x\rangle^{\pm\frac 12}$. In particular, the same proof will establish the absence of eigenvalues for some magnetic fields tending to $+\infty$ or to $0$ at infinity.
\item We will see in Theorem \ref{th-a-bounded} that, when $b$ tends to $0$ too rapidly, the absence of eigenvalues is more subtle to establish.
\end{itemize}
\end{remark}

Our second theorem gives some results in situations where $\vf_+$ and $\vf_-$ are finite but with $\phi = \phi_+ - \phi_- \gg 1$ (the case $\phi \ll -1$ would be similar). By a change of gauge (take $a_0=\int_{-\infty}^{0} b(u)\mathrm{d}u$ in \eqref{def-a}), we can assume that $\phi_-=0$ (and hence $\vf_+ > 0$). 

The problem can then be rewritten in a semiclassical framework. If we set $h = \vf_+^{-1}$, $b_1(x) = hb(x)$ and $a_1(x) = h a(x)$, then we have $\LL = \vf_+^2 \LL_h$ where 
\begin{equation} \label{def-LLh}
\LL_h =h^2D_x^2 + \big( hD_y - a_1(x) \big)^2\,.
\end{equation}
Thus our purpose is now to prove the absence of eigenvalues of the operator $\LL_h$ with
\[
a_1(x) = \int_{-\infty}^x b_1(s) \, \mathrm d s, \quad \int_{-\infty}^{+\infty} b_1(s) \, \mathrm d s = 1.
\]
\begin{theorem} \label{th-a-bounded}
\begin{enumerate}[\rm (i)]
\item \label{th-a-bounded-item-i} For all $h>0$, the operator $\LL_h$ has no point spectrum in $\big[\frac{1}{4},+\infty\big)$.
\item \label{th-a-bounded-item-ii} Assume that $b_1$ is of class $\mathscr{C}^1(\R)$ and takes positive values. Assume also that 
\begin{itemize}
\item for some $N \geq 0$ we have 
\[
b'(x) \underset{\abs x \to +\infty}= \mathscr O \big(\abs x^N\big), 
\]
\item $a_1 \in L^1(\R_-)$ and $(a_1 - 1) \in L^1(\R_+)$.
\end{itemize}
Let $(\etah)_{h>0}$ be such that $\etah = o(\abs{\ln(h)}^{-6})$ as $h \to 0$. Then, there exists $h_0 > 0$ such that for $ h\in (0, h_0)$ the operator $\LL_h$ has no eigenvalue smaller than $\etah$.
\end{enumerate}
\end{theorem}

\begin{remark}
For example, we can apply Theorem \ref{th-a-bounded} to $b_1(x)=\frac{1}{\sqrt{\pi}} e^{-x^2}$. An interesting question is left open: for $h$ small enough, can we exclude the presence of eigenvalues in the interval $\left(\etah,\frac{1}{4}\right)$? One will see in the proof that this function $\etah$ is related to the harmonic approximation. To replace, for instance, $\etah$ by $\frac 14-\varepsilon$ would not only suppose to find a convenient effective Hamiltonian in the harmonic approximation (what is possible via a Birkhoff normal form in dimension one, under analyticity assumptions), but also to be able to deduce from it a non-trivial behavior of each dispersion curve. Even if such a description were possible, it would still not exclude the existence of embedded eigenvalues near $\frac{1}{4}$ in the limit $h\to 0$.
\end{remark}

\subsection{Organization of the proofs}

In Section \ref{sec.fiber}, we recall basic facts about the Fourier fibration of translationnaly invariant magnetic Laplacians. In particular, Proposition \ref{prop-Sigma-lambda} provides a criterion to exclude the existence of eigenvalues as soon as no dispersion curve is constant. Even though this proposition seems to be well-known, the presence of essential spectrum for the fibered operator requires to give a careful proof. This will immediately imply Theorem \ref{th-semi-confined}. Section \ref{sec.harmonic} is devoted to some facts about a parameter dependent version of the harmonic approximation which will be crucial in the proof of Theorem \ref{th-a-bounded} \eqref{th-a-bounded-item-ii} and which will appear when analysing the large frequency limit of the dispersion curves. This approximation will allow us to use somehow the existence of a non-constant \enquote{center-guide dynamics} to prove the non-constant character of some dispersion curves (see Remark \ref{rem.center}).
\newpage
\section{Reminders on fibered magnetic Hamiltonians}\label{sec.fiber}

Since $\mathscr{L}$ commutes with the translation in $y$, the Fourier transform in $y$ will play a fundamental role in our analysis. For $u \in L^2(\R^2)$ and for almost all $x \in \R$ we denote by $u_\x$ the Fourier transform of $u(x,\cdot)$. For $u \in \Sc(\R^2)$ it is given by 
\[
u_\x(x,\x) = \frac{1}{\sqrt{2\pi}}\int_{\R}e^{-iy\xi} u(x,y)\mathrm{d}y\,.
\]
This induces the following direct integral representation (see, for instance, \cite[Section XIII.16]{rs4} about such direct integrals)
\begin{equation} \label{eq-direct-integral}
\mathscr{L}=\int_{\R}^\oplus \mathscr{L}_{\xi} \, \mathrm{d}\xi\,,
\end{equation}
where, for all $\xi\in\R$, 
\[\mathscr{L}_{\xi}=D^2_{x}+(\xi-a(x))^2\,.\]
For all $\xi \in \R$ this defines an operator on $L^2(\R)$ with domain
\[\begin{split}
\mathrm{Dom}(\LL_{\xi})&=\set{u\in H^1(\R) \st a(x)u\in L^2(\R)\mbox{ and }(D_{x}^2+(\xi-a(x))^2)u\in L^2(\R)}\\
&=\set{u\in H^2(\R) \st (\xi-a(x))^2 u\in L^2(\R)}\,.
\end{split}
\]

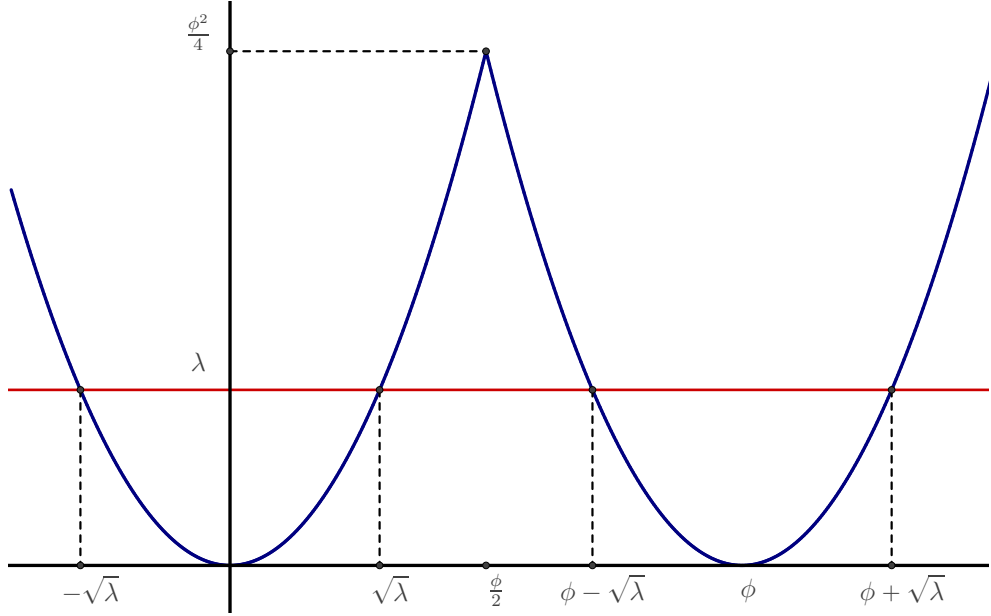
In the following proposition we gather some spectral properties of $\LL_\xi$ that will be useful to the spectral analysis of $\LL$. Let us emphasize here that, in \cite[Assumption (B)]{I85}, the assumption on $b$ implies that $\s_\ess(\mathscr{L}_\xi)=\emptyset$. This will not always be the case in this paper (see Figure \ref{fig.0} where the bottom of the essential spectrum is represented as a function of $\xi$).

\begin{proposition} \label{prop-LLxi}
The operator $\LL_\x$ is self-adjoint and non-negative for all $\x \in \R$. The family $(\LL_\x)_{\x \in \R}$ is analytic of type (A). Let $\xi \in \R$.
\begin{enumerate}[\rm (i)]
\item We have 
\[
\s(\LL_\xi) \subset \Big[ \inf_{x\in\R} \big(\x-a(x)\big)^2  , +\infty \Big).
\]
\item We have 
\[
\s_\ess(\LL_\xi) = \big[ \min \big( (\x-\phi_-)^2 , (\xi - \phi_+)^2 \big) , +\infty \big).
\]
In particular, when $\abs{\phi_-} = \abs{\phi_+} = +\infty$, $\s_\ess(\LL_\xi) = \emptyset$.
\item\label{item-noemb} If $\phi_\pm \in \R$ we assume that $(\xi-a(x))^2 - (\xi-\phi_\pm)^2 \in L^1(\R_\pm)$. Then the operator $\LL_\xi$ has no embedded eigenvalue in $\s_\ess(\LL_\xi)$.
\item The eigenvalues of $\LL_\xi$ are simple and depend analytically on $\xi$.
\end{enumerate}
\end{proposition}

\begin{proof}
The first statements are standard. For (ii), if $\phi_-$ and $\phi_+$ are infinite then $\LL_\xi$ has a compact resolvent by the Riesz-Fr\'echet-Kolmogorov Theorem. If $\phi_-$ and $\phi_+$ are finite then $\LL_\xi$ is a relatively compact perturbation of $D_x^2 + V(x)$ where $V(x) = (\xi-\phi_-)^2$ for $x \leq 0$ and $V(x) = (\xi-\phi_+)^2$ for $x > 0$. We conclude with the Weyl Theorem. If $\phi_- \in \R$ and $\phi_+ = +\infty$ we conclude similarly by considering $V(x) = (\xi-\phi_-)^2$ if $x \leq 0$ and $V(x) = \max((\xi-\phi_-)^2, (\xi-a(x))^2)$ if $x > 0$. The other cases are similar. 

For (iii) we use Lemma \ref{lem.scat}. If $\phi_+ \in \R$ then for $\l \geq (\xi-\phi_+)^2$ we apply the lemma with $\o^2 = \l - (\xi-\phi_+)^2$ and $w = (\x-a(x))^2 - (\xi-\vf_+)^2 \in L^1(\R_+)$. This proves that $\l$ is not an eigenvalue. Similarly, if $\phi_-$ is finite then $\LL_\xi$ has no eigenvalue $\l \geq (\xi-\vf_-)^2$. 

Let us briefly recall why the eigenvalues of $\LL_\xi$ are simple. Assume that $u$ and $v$ are eigenfunctions of $\LL_{\xi}$ associated with the same eigenvalue $\lambda$. Letting $W=u_{1}u'_{2}-u'_{1}u_{2}$, we easily get $W'=0$, so that $W$ is constant. Since $u_{1}$ and $u_{2}$ belong to the domain, we get that $W$ is integrable, and thus that $W=0$. This shows that the family $(u_{1},u_{2})$ is not free. Combining the simplicity of the eigenvalues and the analyticity of the family, we finally get the analyticity of the eigenvalues.
\end{proof}

Let $\xi \in \R$. If $\LL_\xi$ has eigenvalues (necessarily simple and under the essential spectrum, according to Proposition \ref{prop-LLxi}), we label them by increasing order
\[
(\l_k(\xi)) _{1 \leq k \leq N_\xi}, \quad \text{with} \quad  \l_{k}(\xi) < \l_{k+1}(\xi), \quad 1 \leq k < N_\xi,
\]
for some $N_\xi \in \N \cup \set{+\infty}$.

When $\s_\ess(\mathscr{L}_\xi)=\emptyset$, the following proposition can be found in \cite[Theorem XIII.86]{rs4}. In this paper, the essential spectrum will not be empty in general. 

\begin{proposition} \label{prop-Sigma-lambda}
Let $\l \in \R$ and 
\begin{equation} \label{def-Sigma-lambda}
\Sigma_\l = \set{\x \in \R \st \l \notin \s_\ess(\LL_\xi)}.
\end{equation}
If $\l$ is an eigenvalue of $\LL$, then there exists $n \in \N^*$ and a connected component $I$ of $\Sigma_\l$ such that $\LL_\xi$ has at least $n$ eigenvalues for all $\xi \in I$ and
\[
\forall \xi \in I, \quad \l_n(\xi) = \l.
\]
\end{proposition}

\begin{proof}
Let $\l \in \R$ and $u \in \Dom(\LL)\setminus\{0\}$ be such that $\LL u = \l u$. For almost all $\xi\in\R$, we have
\[\LL_{\xi}u_{\xi}=\lambda u_{\xi}\,.\]
Consider
\[\Xi=\set{\xi\in\R : u_{\xi}\neq 0 }\,.\]
In particular, $\l$ is an eigenvalue of $\LL_\xi$ for all $\xi \in \Xi$, and hence, with Proposition \ref{prop-LLxi}, $\Xi \subset \Sigma_\l$. Moreover, $\Xi$ has positive Lebesgue measure, so there exist a connected component $I$ of $\Sigma_\l$ and a compact $K \subset I$ such that $K \cap \Xi$ has positive measure. Then there exists $\xi_0 \in K \cap \Xi$ such that $[\xi_0-\e,\xi_0+\e] \cap \Xi$ has positive measure for all $\e > 0$. Since $\xi_0 \in \Xi$, $\l$ is an eigenvalue of $\LL_{\xi_0}$, so there exists $n \in \N^*$ such that $\LL_{\xi_0}$ has at least $n$ eigenvalues and $\l_n(\xi_0) = \l$. By simplicity of the eigenvalues and continuity with respect to $\xi$, together with the non-negativeness of $\LL_\xi$ (so that the eigenvalues cannot escape to $-\infty$), there exists $\e > 0$ such that $\s(\LL_\xi) \cap [\l-\e,\l+\e] = \set{\l_n(\x)}$ for all $\xi \in [\xi_0 - \e , \xi_0 + \e]$. Since $\l_n$ is analytic and $\l_n(\x) = \l$ on a subset of $[\xi_0-\e,\xi_0 + \e]$ of positive measure, we have $\l_n(\x) = \l$ for all $\x \in [\xi_0-\e,\xi_0 + \e]$.

Assume by contradiction that there exists $\x \in I$ such that $\x > \x_0$ and $\LL_{\xi}$ does not have $n$ eigenvalues. Let
\[
\xi_1 = \sup \set{\xi \in I \st \LL_\xi \text{ has at least $n$ eigenvalues}} \quad \in I.
\]

By analycity we have $\l_n(\x) = \l$ for all $\xi \in [\xi_0,\xi_1)$. Moreover $[\xi_0,\xi_1]$ is a compact subset of $\Sigma_\l$, so $\l < \inf_{\x \in [\xi_0,\xi_1]} \inf \s_\ess(\LL_\xi)$. By continuity of the spectrum of $\LL_\xi$ around $\xi = \xi_1$ we obtain that $\LL_\xi$ has at least $n$ eigenvalues for $\xi$ on some neighborhood of $\xi_1$, which gives a contradiction. Then $\LL_\xi$ has at least $n$ eigenvalues for all $\xi \in I$ with $\xi \geq \xi_0$. The case $\xi \leq \xi_0$ is similar. Then $\l_n$ is defined on the whole interval $I$ and, by analycity, we have $\l_n(\xi) = \l$ for all $\xi \in I$.
\end{proof}
Note that, with these properties in hand, we can easily deduce Theorem \ref{th-conf-semibounded} \eqref{eq.thm1i}: the essential spectrum of $\mathscr{L}_\xi$ is empty so if $\l \in \R$ is an eigenvalue of $\LL$ there exists $n \in \N^*$ such that $\l_n(\xi) = \l$ for all $\xi \in \R$, which is impossible since the bottom of the spectrum of $\LL_\xi$ goes to $+\infty$ when $\xi\to\pm\infty$.

We can also easily prove the first statement of Theorem \ref{th-a-bounded}:

\begin{proof} [Proof of Theorem \ref{th-a-bounded}.\eqref{th-a-bounded-item-i}]
Note that, here, $h>0$ is fixed (and we may assume that $h=1$). 

Assume by contradiction that $\l \geq \frac {1} 4$ is an eigenvalue of $\LL$. By Proposition \ref{prop-LLxi} we have 
\[
\Sigma_\l = \R \setminus \big[ -\sqrt \l , 1 + \sqrt \l \big].
\]
Since $a$ is bounded, we have 
\[
\inf \s(\LL_\x) \limt \x {\pm \infty} + \infty.
\]
Then Proposition \ref{prop-Sigma-lambda} gives a contradiction.
\end{proof}

We cannot use the same argument when $a$ is surjective (since then we have $\inf_{x \in \R} (\xi-a(x))^2 = 0$ for all $\xi \in \R$) or when $a$ is bounded and $\l < \frac {1}{4}$ (because $\Sigma_\l$ has also a bounded connected component, see Figure \ref{fig.0}). 

To go further, we will use the harmonic approximation to estimate the eigenvalues of $\LL_\xi$.

\input figure.tex

\section{Harmonic approximation for moderately small eigenvalues}\label{sec.harmonic}

In this section, we prove a parameter dependent version of the classical harmonic approximation (see for instance \cite{S83, helffer1}). The main interest of Theorem \ref{th-hamonic-approx} below is that we consider eigenvalues which are ``not too small'' (in particular, much larger than the low lying eigenvalues, which are of order $\mathscr{O}(h)$).

Without this version of the harmonic approximation, one would only be able to prove the absence of eigenvalues below $Ch$ in Theorem \ref{th-a-bounded}.

We consider a family $(V_\th)_{\th \in \Th}$ of continuous and real-valued potentials on $\R$ which satisfies the following properties.
\begin{enumerate}[\rm (i)]
\item We can write
\[
V_\th(s) = s^2 v_\th^2 + s^3 w_\th(s)
\]
where, for some $v_-,v_+,C_w,N > 0$, we have 
\begin{equation} \label{eq-estim-v-w}
\forall \th \in \Th,\forall s \in \R, \quad v_- < v_\th < v_+ \quad \text{and} \quad \abs {w_\th(s)} \leq C_w \pppg s^N.
\end{equation}
In particular, there exists $\e_0 > 0$ such that
\[
\forall \th \in \Th, \forall s \in [-\e_0,\e_0], \quad  V_\th(s) \geq \frac {v_- s^2} 2.
\]
\item 
There exists $c_\infty > 0$ such that for $\th \in \Th$ and $s \in \R \setminus [-\e_0,\e_0]$ we have 
\begin{equation} \label{eq-minor-cinfty}
V_\th(s) > c_\infty
\end{equation}
\end{enumerate}

Then, for $h \in (0,1]$ and $\th \in \Th$, we consider the operator 
\[
\LL_{h,\th} = h^2 D_s^2 + V_\th(s),
\]
with domain
\[
\Dom(\LL_{h,\th}) = \set{u \in H^2(\R) \st V_\th u \in L^2(\R)}.
\]

We recall that, for $h \in (0,1]$, the spectrum of the operator $h^2 D_s + v_\th s^2$ is given by the sequence of simple eigenvalues $(2n-1) h v_\th$, $n \in \N^*$. We prove that for $h$ small enough the bottom of the spectrum $\LL_{h,\th}$ is given by simple eigenvalues close to those of this harmonic oscillator.

For $\th \in \Th$ and $h > 0$ we denote by 
\[
 0 <\l_1(h,\th) \leq \l_2(h,\th) \leq \cdots
\]
the eigenvalues of $\LL_{h,\th}$ under the essential spectrum, and we consider a corresponding orthonormal family $(\psi_{k,h,\th})$ of eigenvectors. Then for $E \in (0, \inf \Sp_\ess(\LL_{\th,h}))$ we denote by $N(E,h,\th)$ the number of eigenvalues of $\LL_{h,\th}$ (counted with multiplicities) smaller than $E$:
\[
N(E,h,\th) = \mathsf{max} \set{n \in \N^* \st \l_n(h,\th) \leq E}.
\]
For $\th \in \Th$, $h > 0$ and $n \in \N^*$ we set 
\[
\mathscr E_n(h,\th) = \mathsf{span} (\psi_{k,h,\th})_{1 \leq k \leq n}.
\]

We consider a family $(\etah)_{h > 0}$ of positive numbers such that 
\[
\etah \underset{h \to 0}{=} o \left( \frac 1 {\abs{\ln(h)}^6} \right).
\]

\begin{theorem} \label{th-hamonic-approx}
There exist $h_0> 0$ such that for $\th \in \Th$ and $h \in (0,h_0]$ we have $\etah < \inf \Sp_\ess(\LL_{h,\th})$ and
\begin{equation} \label{estim-N}
N(\etah,h,\th) \geq  \frac {\etah}{4v_+h} - 1\,. 
\end{equation}
Moreover, there exists a function $\e(h)$ converging to $0$ as $h \to 0$ such that, for all $n\in \{1,\dots, N(\etah,h,\th)\}$,
\begin{equation} \label{eq-harmonic-approx}
\abs{\l_n(h,\th) - (2n-1) h v_\th} \leq \e(h) \l_n(h,\th)\,.
\end{equation}
\end{theorem}

\begin{remark}
From \eqref{eq-harmonic-approx} we obtain that for $h$ small enough we have $\l_n(h,\th) \lesssim (2n-1) h v_\th$, so with a possibly different function $\e$ we can rewrite \eqref{eq-harmonic-approx} as 
\begin{equation} \label{eq-harmonic-approx-2}
\abs{\l_n(h,\th) - (2n-1) h v_\th} \leq \e(h) (2n-1) h.
\end{equation}
\end{remark}

The proof of Theorem \ref{th-hamonic-approx} relies on the classical Agmon Formula (see for instance \cite[Prop. 4.7]{raymond}):

\begin{proposition} \label{prop-Agmon}
Let $\Phi$ be a real-valued, Lipschitzian and bounded function on $\R$. Then for $h > 0$, $\th \in \Th$ and $u \in \Dom(\LL_{h,\th})$ we have 
\[
\int_\R \abs{hD(e^{\Phi} u)}^2 \, \mathrm{d}\s + \int_\R \big( V_\th - h^2|\Phi'|^2  \big) e^{2\Phi} \abs u^2 \, \mathrm{d}\s = \Re \innp{\LL_{h,\th}u}{e^{2\Phi}u}.
\] 
In particular, if $(\l,u)$ is an eigenpair of $\LL_{h,\th}$ then 
\[
\int_\R \abs{hD(e^{\Phi} u)}^2 \, \mathrm{d}\s + \int_\R \big( V_\th - h^2|\Phi'|^2  - \l \big) e^{2\Phi} \abs u^2 \, \mathrm{d}\s = 0.
\]
\end{proposition}

On the other hand, the following lemma is an easy consequence of Proposition 4.4 in \cite{raymond}, where we check that the rest is estimated uniformly in $\th \in \Th$.

\begin{lemma} \label{lem-inf-Sp}
There exists $h_0 > 0$ such that for all $\th \in \Th$ and $h \in (0,h_0)$ we have 
\[
\inf \Sp(\LL_{h,\th}) \geq \frac {h v_\th} 2.
\]
\end{lemma}

The following result about the uniform exponential decay of the eigenfunctions has its own interest:

\begin{proposition}\label{prop-exp-decay}
Let $h_0 > 0$ be as in Lemma \ref{lem-inf-Sp}. For
\[
 E  \in \Big( 0, \liminf_{\abs x \to +\infty} \inf_{\th \in \Th} V_\th(x) \Big)
\]
there exist $\gamma>0$ and $C > 0$ such that for $h\in(0,h_{0})$, $\th \in \Th$ and an eigenpair $(\l,\p)$ of $\LL_{h,\th}$ with $\l \leq E$ we have
\[\int_{\R}e^{2\g|s|/\sqrt{\lambda}}|\psi|^2\mathrm{d}s\leq C\|\psi\|_{L^2(\R)}^2\,.\]
\end{proposition}

\begin{proof}
There exist $\kappa \in (0,1)$ and $c_E > 0$ such that for all $\th \in \Th$ and $s \in \R$ we have 
\begin{equation} \label{eq-min-Vth}
V_\th(s) \geq \min \big( c_E s^2 , (1+2\k) E \big).
\end{equation}
Then we set 
\[
\g = \frac {v_- \sqrt \k}{2} > 0,
\]
where $v_-$ is given by \eqref{eq-estim-v-w}. Let $\th \in \Th$ and $h \in (0,h_0)$. Let $(\l,\psi)$ be an eigenpair of $\LL_{h,\th}$ with $\l \leq E$. For $\e > 0$ and $s \in \R$ we set 
\[\Phi_{\varepsilon} (s) = \min \left( \frac {\g |s|}{\sqrt{\lambda}}, \frac 1 \varepsilon \right).\] 
Proposition \ref{prop-Agmon} gives
\[\int_\R \left(V_\theta(s)- \frac {h^2\g^2}{\lambda}-\lambda \right)e^{2\Phi_\varepsilon}|\psi|^2\mathrm{d}s\leq 0\,.\]
By Lemma \ref{lem-inf-Sp} we have $\lambda\geq\frac{hv_-}{2}$, so
\[
\int_\R \big(V_\theta(s)-(1+\k)\l\big)e^{2\Phi_\varepsilon}|\psi|^2\mathrm{d}s
\leq 0\,.
\]
We choose $R>0$ so large that $c_ER^2-(1+\k) \geq \k$. Then we write
\[\int_{|s|\geq R\sqrt{\lambda}} (V_\theta(s)-(1+\kappa)\lambda)e^{2\Phi_\varepsilon}|\psi|^2\mathrm{d}s\leq -\int_{|s| < R\sqrt{\lambda}} (V_\theta(s)-(1+\k)\lambda)e^{2\Phi_\varepsilon}|\psi|^2\mathrm{d}s\,.\]
There exists $c_+ > 0$ such that $0 \leq V_\th(s) \leq c_+ s^2$ for all $\th \in \Th$ and $\abs s \leq R \sqrt E$, so with \eqref{eq-min-Vth} we have
\[\k \l \int_{|s|\geq R\sqrt{\lambda}} e^{2\Phi_\varepsilon}|\psi|^2\mathrm{d}s\leq  \l  (c_+ R^2 + 1 + \k)\int_{|s| < R\sqrt{\lambda}} e^{2\Phi_\varepsilon}|\psi|^2\mathrm{d}s\,,\]
and hence 
\[
\int_{\R} e^{2\Phi_\varepsilon}|\psi|^2\mathrm{d}s\leq \big(1 + \k\inv(c_+ R^2 + 1 + \k) \big)e^{2\g R^2} \int_{|s| < R\sqrt{\lambda}} |\psi|^2\mathrm{d}s.
\]
It only remains to let $\e$ go to 0 to conclude.
\end{proof}

Now we can prove Theorem \ref{th-hamonic-approx}.

\begin{proof}[Proof of Theorem \ref{th-hamonic-approx}]
There exists $h_0 \in (0,1]$ such that Lemma \ref{lem-inf-Sp} holds and for all $h \in (0,h_0]$ and $\th \in \Th$ we have 
\[
\etah \leq \frac {c_\infty} 2 < c_\infty  \leq \inf \Sp_\ess(\LL_{h,\th}),
\]
where $c_\infty$ is given by \eqref{eq-minor-cinfty}.

Let $h \in (0,h_0]$ and $n\in\{1,\ldots, N(\etah,h,\theta)\}$. For $\psi\in \mathscr{E}_{n}(h,\theta)$ with $\nr{\psi}_{L^2(\R)} = 1$ we have
\begin{equation}\label{eq.ubQl}
\langle\LL_{h,\th}\psi,\psi\rangle\leq \lambda_{n}(h,\th)\,.
\end{equation}
On the other hand, by \eqref{eq-estim-v-w},
\begin{equation}\label{eq.Tayl3}
\langle\LL_{h,\th}\psi,\psi\rangle\geq\langle(h^2D^2_{s}+v^2_{\th}s^2)\psi,\psi\rangle-C_w\int_\R |s|^3\langle s\rangle^N|\psi|^2\mathrm{d}s\,.
\end{equation}
Let $\g$ be given by Proposition \ref{prop-exp-decay} for $E = c_\infty / 2$. We set 
\[
\a_n(h,\th) = \frac 2 \gamma \sqrt {\l_n(h,\th)}  \abs{\ln(h)}.
\]
Since $\a_n(h,\th)$ is bounded uniformly in $\th \in \Th$, $h \in (0,h_0]$ and $n \leq N(\etah,h,\th)$, we have
\begin{equation} \label{eq-int-alpha-1}
\int_{\abs s \leq \a_n(h,\th)} |s|^3\langle s\rangle^N|\psi|^2\mathrm{d}s \lesssim \lambda_n(h,\theta)^{\frac{3}{2}} \abs{\ln(h)}^3 .
\end{equation}
Then we consider $c_1,\dots,c_n \in \C$ such that
\[\psi=\sum_{j=1}^n c_j\psi_{j,h,\theta}\,.\]
By the triangle inequality and Proposition \ref{prop-exp-decay} we have
\begin{equation} \label{eq-int-alpha-2}
\begin{aligned}
\big\||s|^\frac{3}{2}\langle s\rangle^{\frac{N}{2}}\psi\big\|_{L^2(|s|\geq \a(h,\th))}
& \leq \sum_{j=1}^n |c_j| \big\||s|^\frac{3}{2}\langle s\rangle^{\frac{N}{2}}\psi_{j,h,\theta}\big\|_{L^2(|s|\geq \a(h,\th))}\\
& \lesssim e^{-\frac {\g \a(h,\th)}{2 \sqrt {\l_n(h,\th)}}} \sum_{j=1}^n |c_j| \Big\|e^{\frac {\g \abs s}{\sqrt {\l_n(h,\th)}}} \psi_{j,h,\theta} \Big\|_{L^2(|s|\geq \a(h,\th))}\\
& \lesssim h \sum_{j=1}^n |c_j|\\
& \lesssim h \sqrt n.
\end{aligned}
\end{equation}
With \eqref{eq-int-alpha-1} and \eqref{eq-int-alpha-2} we get, for some $c > 0$ independant of $\th$, $h$ or $n$,
\[
\langle\LL_{h,\th}\psi,\psi\rangle\geq\langle(h^2D^2_{s}+v^2_{\th}s^2)\psi,\psi\rangle-c\left(\lambda_n(h,\theta)^{\frac{3}{2}}\abs{\ln(h)}^3+nh^2\right)\,.
\]
This, with \eqref{eq.ubQl} and the min-max Theorem, implies that
\begin{equation} \label{estim-lambda-n}
\lambda_{n}(h,\theta)\geq (2n-1)hv_\theta-c\left(\lambda_n(h,\theta)^{\frac{3}{2}}\abs{\ln(h)}^3+nh^2\right)\,.
\end{equation}
In particular, if $h_0$ was chosen small enough, there exists $C > 0$ such that, for $h \in (0,h_0]$, $\th \in \Th$ and $n \leq N(\etah,h,\th)$, 
\begin{equation} \label{eq-nh}
nh \leq C \l_n(h,\th)\,.
\end{equation}
Then \eqref{estim-lambda-n} yields
\begin{equation} \label{eq-lamda-n-geq}
(2n-1) h v_\th - \l_n(h,\th) \leq \e_1(h) \l_n(h,\th)\,,
\end{equation}
where 
\[
\e_1(h) = c\left(\lambda_n(h,\theta)^{\frac{1}{2}}\abs{\ln(h)}^3+C h\right) \limt h 0 0\,.
\]

For $n \in \N^*$ we denote by $f_n$ the $n$-th Hermite function. It solves on $\R$
\[
\big( D_\s^2 +\sigma^2 - (2n-1) \big) f_{n}(\s) = 0 \,.\]
Then for $h > 0$, $\th \in \Th$ and $n \in \N^*$ we set 
\[f_{n,h,\th} : s \mapsto h^{-\frac{1}{4}}v^{\frac{1}{4}}_{\th}f_{n} \Big( h^{-\frac{1}{2}}v^{\frac{1}{2}}_{\th}s \Big)\,.\]
We have $\nr{f_{n,h,\th}} = 1$ and 
\[\big( h^2 D^2_{s}+s^2v^2_{\th}-(2n-1)hv_{\th}\big)f_{n,h,\th}(s)=0\,.\]
For $f$ in $\mathsf{span}(f_{j,h,\th})_{1 \leq j \leq n}$ with $\nr{f}_{L^2(\R)}^2 = 1$ we have 
\[
(2n-1)hv_\th \geq \innp{(h^2 D_s + v_\th^2 s^2)f}{f} \geq \innp{\LL_{h,\th} f}{f} - C_w \int_\R \abs s^3 \pppg s^N \abs f ^2 \, ds.
\]
Following the same lines as above we obtain, for some $C_2 > 0$,
\[
C_w \int_\R \abs s^3 \pppg s^N \abs f ^2 \, ds \leq \rho(n,h) := C_2 \big( (nh)^{\frac 32} \abs{\ln(h)}^3 + nh^2 \big).
\]
If $n \in \N^*$ is not greater than $\etah / (4v_+h)$ we have
\[
\frac {\rho(n,h)}{\etah} \leq \frac {C_2}{(4v_+)^{\frac 32}} \etah^{\frac 12} \abs{\ln(h)}^3 + \frac {h}{4v_+} \limt h 0 0.
\]
Hence, if $h_0$ is small enough, then for $h \in (0,h_0]$, $\th \in \Th$ and $n \leq \etah / (4v_+h)$ we have 
\[
\innp{\LL_{h,\theta} f}{f} \leq (2n-1)v_\th h + \rho(n,h) \leq \etah. 
\]
By the min-max Theorem this implies $\l_n(h,\theta) \leq \etah$, and \eqref{estim-N} is proved.

On the other hand for $n \leq N(\etah,h,\theta)$ we have 
\begin{equation} \label{eq-lamda-n-leq}
\l_n(h,\th) - (2n-1)hv_\th \leq \l_n(h,\th) \e_2(h),
\end{equation}
where, by \eqref{eq-nh}, 
\[
\e_2(h) := \sup_{n \leq N(\etah,h,\theta)} \frac {\rho(n,h)}{\l_2(n,h)} \leq C_2 C^{\frac 32} \etah^{\frac 12} \abs{\ln(h)}^3 + C_2 C h \limt h 0 0.
\]
Then  \eqref{eq-harmonic-approx} follows from \eqref{eq-lamda-n-geq} and \eqref{eq-lamda-n-leq}. 

\end{proof}

\section{Absence of embedded eigenvalues with transverse confinement}

In this section, we prove Theorem \ref{th-semi-confined}. 

Since $\phi_+=+\infty$, we observe that if $\phi_- = +\infty$, we can apply Theorem \ref{th-conf-semibounded}. Thus, we can restrict our attention to the cases $\phi_- \in \R$ and $\phi_- = -\infty$. The proof relies on the following asymptotics for the eigenvalues:

\begin{proposition}\label{prop.nsb}
Assume that \eqref{hyp-a-asymp} holds (for any $\a > -1$)  and that $\phi_- \in [-\infty, +\infty)$. 
Let $n \in \N^*$. Then for $\xi$ large enough the operator $\LL_\xi$ has at least $n$ eigenvalues and its $n$-th eigenvalue $\l_n(\xi)$ satisfies 
\[\lambda_{n}(\xi)\underset{\xi\to+\infty}{=}(2n-1)c_{1}c_{0}^{-\frac{\alpha}{1+\alpha}}\xi^{\frac{\alpha}{1+\alpha}}+o(\xi^{\frac{\alpha}{1+\alpha}})\,,\]
where $c_0 = c_1 / (1+\alpha)$.
\end{proposition}

\begin{proof}
There exists $x_0 \geq 1$ such that for $x \geq x_0$ we have 
\begin{equation} \label{eq-minor-der-a}
\quad a'(x) = b(x) \geq \frac {c_1  x^{\a}}2.
\end{equation}
In particular, $a$ is increasing on $[x_0,+\infty)$. Since $a$ has a limit in $[-\infty,+\infty)$ at $-\infty$ we can assume, by choosing $x_0$ larger if necessary, that $a(x_0) > a(x)$ for all $x \in (-\infty,x_0)$. We set $\x_0 = a(2x_0)$. Then for $\x \geq \x_0$ there is a unique $x_\xi \in \R$ such that $a(x_\xi) = \xi$. Since
\[
a(x)=\int_0^{x}b(u)\,\mathrm{d}u\underset{x\to+\infty}{\sim} c_{0} x^{\alpha+1}\,,
\]
it satisfies
\[
x_{\xi}\underset{\xi\to+\infty}{\sim}(c_{0}^{-1}\xi)^{\frac{1}{1+\alpha}}\,.
\]
Let $\x \geq \x_0$. For $v \in L^2(\R)$ and $s \in \R$, we set 
\[
(U_\x v)(s) =  x_\x^{\frac 12}  v \big( x_\x (1+s) \big)\,. 
\]
$U_\x$ is a unitary operator on $L^2(\R)$ and 
\begin{equation} \label{eq-U-xi}
U_\xi \LL_\x U_\xi\inv  = x_\x^{-2} D_s^2  +  \big( \x - a(x_\x (1+s) \big)^2 =\xi^2\left[h_\xi^2 D_s^2 + V_\xi(s)\right]\,,
\end{equation}
where 
\[
h_\xi = (\xi x_\xi) \inv \quad \text{and} \quad V_\xi(s) = \big(1 -\xi^{-1} a(x_\x(1+s)\big)^2\,.
\]
$V_\xi$ takes non-negative values and has a unique zero at $s = 0$. By the Taylor formula, 
\begin{equation} \label{eq-taylor-a}
a \big( x_\x (1+s) \big) = \x + s x_\x a'(x_\x) +  s^2 x_\x^2 \int_0^1 (1-\tau) a'' \big( (1+\tau s) x_\x \big) \mathrm{d}\tau\,,
\end{equation}
so we can write 
\[
V_\xi(s) = s^2 v_\xi^2 + s^3 w_\xi(s)\,,
\]
where, by using \eqref{hyp-a-asymp},
\[v_\xi = \xi^{-1}x_\xi a'(x_\xi)\underset{\xi\to+\infty}{\to}\frac{c_{1}}{c_{0}} \quad \text{and} \quad 
\abs{w_\x(s)} \leq \tilde C \langle s\rangle^{ \max(2\alpha-1,\alpha-1)}\,,
\]
for some $\tilde C > 0$ independent of $\xi$ and $s$.

Let us now consider the coercivity property away from the minimum. Let $\e \in \big( 0,\frac 12 \big)$. Let $\x \geq \x_0$. For $s \geq \e$ we have by the Mean Value Theorem and \eqref{eq-minor-der-a}
\[
a(x_\xi(1+s)) - a(x_\xi) \geq \frac {c_1 s x_\xi^{\a +1}} 2  \geq c_\e \x, 
\]
for some $c_\e > 0$, and hence 
\[
V_\x(s) \geq c_\e^2.
\]

Similarly, if $s \leq -\e$ we have 
\[
a(x_\xi) - a(x_\xi(1+s)) \geq a(x_\xi) - a(x_\xi(1-\e))  \geq \frac {c_1 \e x_\x} 2 \left( \frac {x_\xi} 2 \right)^{\a},
\]
and we conclude similarly. In any case we obtain $c_\infty > 0$ such that for $\xi\geq \xi_0$ and $\abs s \geq \e$ we have 
\[V_{\xi}(s)\geq c_\infty\,.\]
With all these properties we can apply Theorem \ref{th-hamonic-approx}. We obtain that, for all $n \in \N^*$, there exists $\x_0 \geq 0$ such that for $\x \geq \x_0$ the operator $h_\xi^2 D_s^2 + V_\xi(s)$ has at least $n$ eigenvalues and its $n$-th eigenvalue $\tilde \l_n (\x)$ satisfies
\[
\tilde \l_n(\x)  \underset{\xi\to+\infty}{\sim} (2n-1) h_\x v_\x \underset{\xi\to+\infty}{\sim}(2n-1)c_{1}c_{0}^{-\frac{\alpha}{1+\alpha}}\xi^{-\frac{1}{1+\alpha}-1}\,.
\]
The asymptotic behavior of $\l_n(\x)$ follows since, by \eqref{eq-U-xi}, we have $\l_n(\xi) = \xi^2 \tilde \l_n(\xi)$.
\end{proof}

Now we can prove Theorem \ref{th-semi-confined}.

\begin{proof}[Proof of Theorem \ref{th-semi-confined}]
Let $\l \geq 0$ and assume by contradiction that $\l$ is an eigenvalue of $\LL$.

Consider the case $\phi_- = -\infty$. Then, for all $\xi \in \R$ the spectrum of $\LL_\xi$ is purely discrete. By Proposition \ref{prop-Sigma-lambda}, there exists $n \in \N^*$ such that $\l_n(\x) = \l$ for all $\xi \in \R$. This gives a contradiction with Proposition \ref{prop.nsb}.

Consider now the case $\phi_- \in \R$. Then, we have $\Sigma_\l = \R \setminus [\phi_- - \sqrt \l , \phi_- + \sqrt \l]$ and we consider its two connected components in order to apply Proposition \ref{prop-Sigma-lambda}.
\begin{itemize}

\item By Proposition \ref{prop.nsb}, $\l$ cannot be an eigenvalue of $\LL_\xi$ for all $\xi > \phi_- + \sqrt \l$. 

\item Since $a$ is bounded from below, Proposition \ref{prop-LLxi} gives $\displaystyle{\lim_{\xi\to-\infty}\inf\s(\mathscr{L}_\xi)=+\infty}$, so $\l$ cannot be an eigenvalue of $\LL_\xi$ for all $\xi<\phi_--\sqrt{\lambda}$.
\end{itemize}
This is a contradiction.
\end{proof}

\begin{remark}
These arguments also imply Theorem \ref{th-conf-semibounded} \eqref{eq.thm1ii}. Since $b_\pm \neq 0$, we are in a situation where $\Sp_\ess(\LL_\xi)$ is empty for all $\xi \in \R$, so if $\LL$ has an eigenvalue there exists $n \in \N$ such that $\l_n(\xi)$ does not depend on $\xi$. This gives a contradiction since, by Proposition \ref{prop.nsb}, we should have 
\[
\lim_{\xi \to \pm \infty} \l_n(\xi) = (2n-1) b_\pm\,.
\]
\end{remark}

\section{Moderately small eigenvalues without transverse confinement}
In this section we prove the second statement of Theorem \ref{th-a-bounded}. We recall that $b_1$, $a_1$ and the operators $\LL_h$, $h > 0$, were defined before the statement of Theorem \ref{th-a-bounded}.

For $\theta \in \R$, we let
\[
\LL_{h,\theta} = h^2D_x^2 +(\theta - a_1(x))^2.
\]
Then $\LL_h$ is the direct integral of $\LL_{h,\theta}$, $\th \in \R$, as in \eqref{eq-direct-integral}.

\begin{proof} [Proof of Theorem \ref{th-a-bounded}.\eqref{th-a-bounded-item-ii}]

Since $b_1$ takes positive values, $a_1$ is an increasing bijection from $\R$ to $(0,1)$. For $\th \in (0,1)$ we set $x_\theta = a_1\inv(\theta)$. Then for $s \in \R$ we set $V_\theta(s) = (\theta-a_1(x_\th+s))^2$. This defines a nonnegative valued potential, $0$ is the unique solution of $V_\theta(0) = 0$ and $V''_\theta(0) = 2 b_1(x_\theta)^2 > 0$, so $V_\theta$ has a unique non-degenerate minimum at $0$ (and this minimum is not attained at infinity).

Let $J$ be a compact interval of $\R$ on which $b$ is not constant and $\Th = a(J)$. As in \eqref{eq-taylor-a} we write 
\[
a_1(x_\th + s) = \th + s b_1(x_\th) + s^2 I(\th,s),
\]
where
\[
I(\theta,s) = \int_0^1 (1-\tau) b_1'(x_\th + s\tau) \, \mathrm d \tau.
\]
 This gives 
\[
V_\th(s) = s^2 b_1(x_\th)^2 + s^3 \big( 2b_1(x_\theta) I(\theta,s) + s I(\theta,s)^2 \big).
\]
Since $b_1$ is continuous and takes postive values, there exist $v_-,v_+ > 0$ such that $v_- < b_1(x_\theta) < v_+$ for all $\th \in \Th$. On the other hand, since $b_1'$ grows at most polynomially, this is also the case for $I(\theta,\cdot)$, uniformly in $\theta \in \Theta$. Thus, we can apply Theorem \ref{th-hamonic-approx}. By \eqref{eq-harmonic-approx-2} there exist $h_0 > 0$ and $\e : \R_+^* \to \R_+^*$ going to 0 at 0 such that for $\theta \in \Th$, $h\in(0,h_0)$ and $n \leq N(\etah,h,\th)$ we have 
\begin{equation} \label{estim-lk-vf}
\abs{\l_n(h,\theta) - (2n-1) h b_1(x_\theta)} \leq  \e(h) (2n-1) h.
\end{equation}

Let $x_1,x_2 \in J$ be such that $b_1(x_1) \neq b_1(x_2)$. We set $\th_1 = a_1(x_1)$, $\th_2 = a_1(x_2)$. Choosing $h_0$ smaller if necessary, we can assume that for all $h\in(0,h_0)$ we have 
\begin{equation} \label{eq-diff-b}
\abs{b_1(x_1) - b_1(x_{2})} > \e(h).
\end{equation}

Now assume by contradiction that there exist $h\in(0,h_0)$ and $\l \in [0,\etah]$ such that $\l$ is an eigenvalue of $\LL_h$. We necessarily have $\l \in \big[0,\frac {1}{4} \big)$. Then, with $\Sigma_\l$ defined as in \eqref{def-Sigma-lambda}, we have 
\[
\Sigma_\l = \big( -\infty, -\sqrt \l \big) \cup \big( \sqrt \l , 1 - \sqrt \l \big) \cup \big( 1+\sqrt \l,+\infty \big).
\]
As in the proof of the first statement of Theorem \ref{th-a-bounded} we see that $\l$ cannot be an eigenvalue of $\LL_{h,\theta}$ for all $\theta \in (-\infty, -\sqrt \l)$ or for all $\th \in (1 + \sqrt \l, +\infty)$, so by Proposition \ref{prop-Sigma-lambda} there exists $k \in \N^*$ such that $\l = \l_k(h,\theta)$ for all $\theta \in \big(\sqrt{\l} , 1 - \sqrt{\l} \big)$. If $h_0$ was chosen small enough, we have $\th_1,\th_2 \in  \big(\sqrt{\l} , 1 - \sqrt{\l} \big)$, so $\l_k(h,\theta_1) = \l = \l_k(h,\theta_2)$, which gives a contradiction with \eqref{estim-lk-vf} and \eqref{eq-diff-b}.
\end{proof}

\begin{remark}
Note that \eqref{estim-lk-vf} describes the dispersion curves on the interval $(\sqrt{\lambda}, 1-\sqrt{\lambda})$, see Figure \ref{fig.0}. The eigenvalues under consideration here are far below the \enquote{peak} of the essential spectrum. 
\end{remark}

\begin{remark}\label{rem.center}
The function $\Theta\ni\theta\mapsto b_1(x_\theta)$ is nothing but an effective Hamiltonian which emerges from the semiclassical limit. In the semiclassical spectral theory of the magnetic Laplacian, this effective Hamiltonian appears, for instance, in \cite[Theorem 1.1]{RVN15}. With this interpretation, the function $\theta\mapsto x_\theta$ corresponds to a parametrization of the \enquote{characteristic manifold} of the magnetic Laplacian.
\end{remark}

\subsection*{Acknowledgments}
This work has been supported by the CIMI Labex, Toulouse, France, under grant ANR-11-LABX-0040-CIMI. N. Raymond is deeply grateful to the Mittag-Leffler Institute where part of this work was completed.

\appendix 

\section{}
The following lemma is very classical and originally appears in \cite{jost}.
\begin{lemma}\label{lem.scat}
Let $\o \geq 0$ and $w \in L^1(\R_{+})$. Let $\psi \in C^2(\R_+)$ be such that 
\begin{equation*}
-\psi''- \o^2 \psi + w \psi=0\,.
\end{equation*}
There exists a unique $(a,b)\in\C^2$ such that 
\begin{eqnarray*}
\psi(x) & \underset{x\to+\infty}{=} & ae^{i\omega x}+be^{-i\omega x}+o(1),\quad \o > 0,\\
\psi(x) & \underset{x\to+\infty}{=} &  a + bx + o(1),\quad \o = 0.
\end{eqnarray*}
In particular, if $\psi\in L^2(\R_{+})$, then $\psi=0$.
\end{lemma}

\begin{proof}
We first assume that $\o > 0$. For $x \geq 0$ we set $U(x) = (\p(x) , \p'(x)) \trsp$. Then $U \in C^1(\R_+)$ and 
\[U'= \begin{pmatrix} 0 & 1 \\ -\o^2 & 0 \end{pmatrix}U+\begin{pmatrix} 0 & 0 \\ w  & 0 \end{pmatrix}U\]
We have
\[
P\inv \begin{pmatrix} 0 & 1 \\ -\o^2 & 0 \end{pmatrix} P = i\Omega\,,\quad  \Omega = \begin{pmatrix} \omega & 0 \\ 0 & -\omega \end{pmatrix}, \quad 
P = \begin{pmatrix} 1 & 1 \\ i\o & -i\o \end{pmatrix}\,.
\]
Then we have 
\[
V'(x) = i\O V(x) + M(x) V(x)
\]
where 
\[
V = P\inv U, \quad \text{and} \quad M = P\inv \begin{pmatrix} 0 & 0 \\ w & 0 \end{pmatrix} P\in L^1(\R_{+})\,.
\]
The Duhamel Formula gives, for all $x\geq 0$,
\begin{equation}\label{eq-Duhamel}
V(x) = e^{i\Omega x}V(0) + \int_0^x e^{i\O (x-s)} M(s) V(s) \mathrm{d}s\,.
\end{equation}
In particular,
\[
\nr{V(x)} \leq\nr{V(0)} + \int_0^x \nr{M(s)} \nr{V(s)} \mathrm{d}s\,,
\]
and hence, by the Gronwall Lemma,
\[
\|V(x)\| \leq \nr{V(0)} e^{\int_0^x \|M(s)\| \mathrm{d}s}\,.
\]
This proves that $V$ is bounded. Thus, by \eqref{eq-Duhamel} we can set 
\[
A = \lim_{x \to +\infty} e^{-i\Omega x} V(x).
\]
The Duhamel Formula now gives 
\begin{equation} \label{eq-Duhamel-2}
V(x)=e^{i\Omega x}A-\int_{x}^{+\infty} e^{i\O (x-s)} M(s) V(s)\mathrm{d}s \underset{x\to+\infty}{=} e^{i\Omega x}A + o(1).
\end{equation}
It remains to multiply by $P$ to conclude. If $\o = 0$ then we proceed similarly, without change of basis, and using the fact that
\[
\exp \begin{pmatrix} 0 & x \\ 0 & 0 \end{pmatrix}  = \begin{pmatrix} 1 & x \\ 0 & 1 \end{pmatrix}. 
\]
This establishes the existence of $a$ and $b$. Since they are necessarily unique, the proof is complete.

\end{proof}

\end{document}